\documentclass[aps,prl,reprint,groupedaddress]{revtex4-1}

\usepackage{amsmath}
\usepackage{amssymb}
\usepackage{graphicx}
\usepackage[thmmarks]{ntheorem}
\hypersetup{colorlinks=true,breaklinks=true}

\DeclareMathOperator*{\rank}{rank}
\DeclareMathOperator*{\sign}{sgn}

\makeatletter
\newsavebox\myboxA
\newsavebox\myboxB
\newlength\mylenA
\newcommand*\xoverline[2][0.75]{%
	\sbox{\myboxA}{$\m@th#2$}%
	\setbox\myboxB\null
	\ht\myboxB=\ht\myboxA%
	\dp\myboxB=\dp\myboxA%
	\wd\myboxB=#1\wd\myboxA
	\sbox\myboxB{$\m@th\overline{\copy\myboxB}$}
	\setlength\mylenA{\the\wd\myboxA}
	\addtolength\mylenA{-\the\wd\myboxB}%
	\ifdim\wd\myboxB<\wd\myboxA%
	\rlap{\hskip 0.5\mylenA\usebox\myboxB}{\usebox\myboxA}%
	\else
	\hskip -0.5\mylenA\rlap{\usebox\myboxA}{\hskip 0.5\mylenA\usebox\myboxB}%
	\fi}
\makeatother

\theoremheaderfont{\itshape}\theorembodyfont{\upshape}
\theoremnumbering{arabic}
\theoremseparator{.}
\theorempreskip{0ex}\theorempostskip{0ex}
\newtheorem{definition}{Definition}
\newtheorem{proposition}{Proposition}
\newtheorem{remark}{Remark}
\newtheorem{corollary}{Corollary}

\theoremheaderfont{\itshape}\theorembodyfont{\upshape}
\theoremstyle{nonumberplain}
\theoremseparator{.}
\theoremsymbol{\ensuremath\square}
\newtheorem{proof}{proof}

\begin{document}
\title{Characterizing Incomparability in Quantum Information}
\date{\today}
\author{Liwen Hu}
\email{huliwen@bit.edu.cn}
\affiliation{School of Physics, Beijing Institute of Technology, Beijing 100081, China}

\begin{abstract}
	The theory of majorization has seen substantial application in quantum information. Its framework predicates on the comparability between real vectors. We explore the antithesis of this premise, namely, \emph{incomparability}. Specifically, we provide ways to measure the incomparability between a pair of spectra. We show that distinct spectra isoentropic by generalized entropies are \emph{mutually} strongly incomparable. The inversion rank is proposed to classify incomparability. Majorization relations are advanced to probe the scale of incomparability, referred to as \emph{inconvertibility}.
\end{abstract}
\maketitle

\textit{Introduction.}\textemdash The theory of majorization~\cite{marshall2010} traces its roots to the characterization of inequalities and has deepseated foundation in matrix theory. In its incipient stages of application, majorization has become instrumental in the explanation of many crucial aspects of quantum physics. Examples include the deterministic conversion of quantum states~\cite{nielsen1999,du2015,zhu2017,gour2015}, the characterization of mixing and measurement~\cite{nielsen2001a}, the detection of separability~\cite{nielsen2001b}, the formulation of uncertainty relations~\cite{partovi2011} and the framework of quantum relative Lorenz curves~\cite{buscemi2017}, etc. Majorization permits only a preorder on real vectors, its utility in large relies on their comparability while cases proving otherwise are deemed undesirable.

The notion of \emph{incomparability} is a prevalent theme in order theory and an ubiquitous phenomenon in nature. First educed by Nielsen~\cite{nielsen2001a}, majorization incomparability was used to account for the restrictions placed on entanglement transformation. Work has been done to uncover its underlying nature~\cite{zyczkowski2002,clifton2002,chattopadhyay2007,chattopadhyay2008}, however such understanding is far from exhaustion. Here, we examine majorization incomparability under a quantum setting in hopes of facilitating such discussions.

As a preliminary, let us outline the basic principles and provide some necessary notations. For $d$-dimensional density matrix $\rho$, its spectrum $\lambda[\rho]$ occupies a $(d-1)$-dimensional spectral simplex $\Delta_{d - 1}$ which can be partitioned into $d!$ minor simplices called \emph{Weyl chambers}, each containing the complete set of spectra but with different order~\cite{zyczkowski2002}. We arrange $\lambda[\rho]$ in nonincreasing order as a vector $r = \lambda^{\downarrow}[\rho] \in \mathcal{W}_{d - 1}$ such that $1 \ge r_i \ge r_{i + 1} \ge 0, i \in \{ 1, \ldots, d - 1 \}$, where $\mathcal{W}_{d - 1}$ is the one with nonincreasing order, which we henceforth merely refer to as the Weyl chamber. It pertains to the scope of our analysis. For such spectral vectors $r$ and $s$, if they satisfy $A_j(r) \equiv \sum_{i=1}^{j} r_i \geq \sum_{i=1}^{j} s_i \equiv A_j(s), j \in \{ 1, \dots, d \}$, with equality when $j=d$, then $r$ majorizes $s$, i.e., $r \succ s$ or $\rho \succ \sigma$. Alternatively, we write $A_j'(r) = 1 - A_{j - 1}(r)$ for summation in nondecreasing order. We say the majorization is \emph{exact} if there exists a coincidence $A_j(r) = A_j(s)$ where $j < d$, i.e., $r \succeq s$. We let $\mathcal{M}_{d - 1}^-(r) \equiv \{ r' \in \mathcal{W}_{d - 1} : r \succ r' \}$ be the set of nonincreasing spectra $r$ majorizes in a $d$-system and let $\mathcal{M}_{d - 1}^+(r) \equiv \{ r' \in \mathcal{W}_{d - 1} : r' \succ r \}$ be the set of those majorizing $r$. There are instances when $r \nsucc s$ and $s \nsucc r$ meaning neither party majorizes the other, nonetheless, if the final equality holds then we say that $r$ and $s$ are incomparable, i.e., $r \nsim s$. We let $\mathcal{I}_{d - 1} \equiv \{ r' \in \mathcal{W}_{d - 1} : r' \nsim r \}$  be the set incomparable to $r$. There are times when $r$ does not necessarily majorize $s$, but do so by appending a catalyst $c$~\cite{jonathan1999}, i.e., $r \otimes c \succ s \otimes c$. We say that $r$ trumps $s$, i.e., $r \succ_T s$. When $r \nsucc_T s$ and $s \nsucc_T r$ they are called strongly incomparable~\cite{bandyopadhyay2002,duan2005}, i.e., $r \nsim_T s$.

We adopt the viewpoint of deterministic state conversion in quantum resource theories of which we highlight three. Regarding entanglement, the bipartite pure state transformation $|\psi \rangle \rightarrow |\phi \rangle$ eventuates via local operation and classical communication iff the Schmidt vector of $|\phi \rangle$ majorizes that of $|\psi \rangle$~\cite{nielsen1999}. Similarly for coherence, the pure state transformation $|\psi \rangle \rightarrow |\phi \rangle$ eventuates via incoherent operation iff the dephasing of $|\phi \rangle$ majorizes that of $|\psi \rangle$~\cite{du2015,zhu2017}. As for the resource theory of nonuniformity, the transformation $\rho \rightarrow \sigma$ eventuates via unital operation iff $\rho \succ \sigma$~\cite{gour2015}. The majorizing party holds less resource for the first two theories, but more resourceful for the last.

When state conversion cannot occur with certainty, especially in the event of strong incomparability, where even assistance does not guarantee success, it is reasonable to inquire the modifications a source or end state has to undergo for assured transformation. Although this may constitute costly operation, our interest diverges from that of resource theories and is not consigned to free operations. If the amount of resource a state possesses, when in excess, is of no immediate consequence to the realization of a task, the lack thereof from a bare minimum is then of more import. Incomparability is a condition to be extricated from, in some sense it implicates a notion of ``anti-resource'' that needs to be expended.

Aside from quantifying incomparability, the strongly incomparable nature of distinct isoentropic spectra is analyzed, the concept of incomparability is refined through categorization by inversion rank, representative spectral families are developed to ultimately survey the scale of incomparability, referred to as \emph{inconvertibility}.

\textit{Measuring incomparability.}\textemdash To gauge the incomparability of a pair of spectra we need to determine the expenditure necessary to render them comparable by a particular means with optimality. By this logic, we demonstrate characteristic ways to quantify incomparability, including approaches ranging operational, distansal and algebraic perspectives. They are designed to be nonnegative, vanishing for comparable pairs and symmetric.

Quantum states tend toward entropy increase under quantum operations. We capture this through measurements that enable majorization.

\begin{definition}
	Let $\rho, \sigma \in \mathcal{B}(\mathcal{H}_d)$, the majorization cost of $\rho$ with respect to $\sigma$ can be expressed as
	\begin{equation}
		M(\rho | \sigma) = \inf \{ S(\Lambda_\pi[\rho]) - S(\rho) : \sigma \succ \Lambda_\pi[\rho] \},
	\end{equation}
	where $S(\rho)$ is the von Neumann entropy, infimum is taken over measurements $\Lambda_\pi[\rho]$ in the bases $\left\lbrace \pi = \{ \Pi_j \}_{i=1}^d \right\rbrace $. Their operational incomparability is the minimum of their majorization costs:
	\begin{equation}
		I_\mathrm{O}(\rho, \sigma) = \min \{ M(\rho | \sigma), M(\sigma | \rho) \}.
	\end{equation}
\end{definition}

$M(\rho | \sigma)$ is simply the minimal entropy cost of measuring $\rho$ for $\sigma$ to majorize it. We know from the Schur-Horn theorem~\cite{marshall2010} that $\mathcal{M}_{d - 1}^-(r)$ is completely contained within the diagonal entries of possible density matrices for $r$, which are obtained by $\Lambda_\pi[\rho]$. Thus, deriving $M(\rho | \sigma)$ and $M(\sigma | \rho)$ translates to finding their optimal posterior spectra $r^\star = s^\star \in \mathcal{M}_{d - 1}^-(r) \cap \mathcal{M}_{d - 1}^-(s)$ which coincide. Wherefore, $I_\mathrm{O}(\rho, \sigma) = S(r^\star) - \max \{S(\rho), S(\sigma)\}$. Hence, by Uhlmann's theorem~\cite{marshall2010} supplanting $\Lambda_\pi[\rho]$ with general bistochastic operations yields the same result.

We present a general framework for distance measures, unlike before the party under optimization may now move up and down the majorization hierarchy.

\begin{definition}
	Let $r, s \in \mathcal{W}_{d-1}$, the ascending and descending majorization distances of $r$ with respect to $s$ can be expressed as
	\begin{align}
		D^\leftarrow(r | s) & = \inf \{ \| r' - r \| : r' \succ s \},\\
		D^\rightarrow(r | s) & = \inf \{ \| r' - r \| : s \succ r' \},
	\end{align}
	where ``$\leftarrow$'' and ``$\rightarrow$'' each stands for $r$ moving up and down the majorization order. Their incomparability distance is the minimum of their majorization distances:
	\begin{equation}
		\begin{array}{cc}
		I_\mathrm{D}(r, s) = \min \{ D^\alpha(r | s), D^\alpha(s | r) \}, & \alpha = \leftarrow, \rightarrow.
		\end{array}
	\end{equation}
\end{definition}

When $D$ is taken to be the trace distance, $D^\leftarrow(s | r)$ corresponds to the maximal fidelity achievable in a faithful entanglement transformation with known form~\cite{vidal2000}. Here, we care not for which one majorizes which, all such distances are found and their infimum serves as the incomparability of a spectral pair.

\begin{figure}
	\centering
	\includegraphics[height=0.130\textwidth]{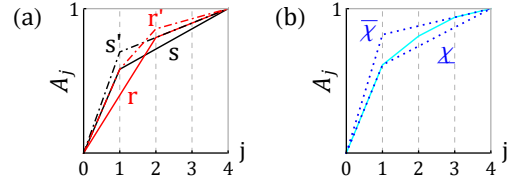}
	\caption{\label{fig:lorenz}(a) Lorenz curves of $r=(0.4, 0.4, 0.1, 0.1)$ (solid red), $s=(0.58, 0.14, 0.14, 0.14)$ (solid black), $r'=(0.58, 0.28, 0.07, 0.07)$ (dotdashed red), $s'=(0.7, 0.1, 0.1, 0.1)$ (dotdashed black). $r \nsim s$. $r' \succeq s$, $s' \succeq r$. (b) Lorenz curves of pure depolarization bounds. $\overline{\chi}=(0.82, 0.06, 0.06, 0.06) \succeq (0.61, 0.20, 0.13, 0.06) \succeq (0.61, 0.13, 0.13, 0.13)=\underline{\chi}$.}
\end{figure}

For the algebraic method we intend to "stretch" the Lorenz curve~\cite{marshall2010,buscemi2017} of $r$ upwards to majorize $s$ and vice versa, whichever gives the least distortion is used to represent incomparability. We know from probabilistic entanglement conversion, the optimal probability for $|\psi \rangle \rightarrow |\phi \rangle$, with respective Schmidt vectors $r$ and $s$, is specified by $Q(s | r) \equiv P(|\psi \rangle \rightarrow |\phi \rangle) = \min_{j \in \{1, \ldots, d\}} A_j'(r) / A_j'(s)$~\cite{vidal1999}. It suits our motive to determine the least \emph{failure rate} for interconversion.

\begin{definition}
	Let $r, s \in \mathcal{W}_{d-1}$, their algebraic incomparability is defined as
	\begin{equation}
		I_\mathrm{A}(r, s) = 1 - \max \{ Q(r | s), Q(s | r) \}.
	\end{equation}
\end{definition}

Ref.~\cite{feng2005} details an intuitive procedure: $r' = (1 - \mu + \mu r_1, \mu r_2, \ldots, \mu r_d), \mu \in [0, 1]$. $r_1$ is enhanced while $r_{i > 1}$ are attenuated. $r' \succeq s$ for $\mu = Q(r | s)$ [FIG.~\ref{fig:example}(a)].

A comparison of the measures is made in FIG.~\ref{fig:example}(b). Operationally, besides $I_\mathrm{O}$, a local variant $I_\mathrm{lO}$ (assuming a $2 \times 2$-system) is supplied for contrast. The former is derived by optimizing entropy in $\mathcal{M}_3^-(\eta) \cap \mathcal{M}_3^-(\xi)$. The latter enlists parameterizing local measurements with Bloch polar angles. Note, $\eta^{\mathrm{l}\star} \neq \xi^{\mathrm{l}\star}$. The Euclidean distance is used for $I_\mathrm{D}$ while $I_\mathrm{A}$ is straightforwardly calculated. $I_\mathrm{O}(\xi, \eta) \leq I_\mathrm{lO}(\xi, \eta)$, indicating limited local access. $I_\mathrm{O}(\tilde{\xi}, \eta)$ is non-smooth as $S(\tilde{\xi}) = S(\eta)$ where the optimal party is reversed. Also, $I_\mathrm{O}(\tilde{\xi}, \eta) \approx I_\mathrm{lO}(\tilde{\xi}, \eta)$. $I_\mathrm{D}$ and $I_\mathrm{A}$ are linear with respect to $\| \xi - \eta \|$.

\textit{Isoentropic spectra.}\textemdash It was shown that isoentropic states are either unitarily connected or incomparable~\cite{chattopadhyay2008,li2013}. This was also proven for unified entropies~\cite{liu2014}. Here, we solidify and advance these claims to the extent of strong incomparability for generalized entropies.

Trumpability is equivalent to a series of inequalities for a family of R\'{e}nyi entropy based functions~\cite{klimesh2007,turgut2007}. Let
\begin{equation}
	f_\nu(r) =
	\begin{cases}
		\frac{1}{\nu (1 - \nu)} \log \sum_i r_i^\nu & \nu \neq 0,1, \\
		\frac{1}{d} \sum_i \log r_i & \nu = 0,\\
		- \sum_i r_i \log r_i & \nu = 1,
	\end{cases}
\end{equation}
where $f_0(r)$ is the Burg entropy and $f_1(r)$ is the Shannon entropy. Let $r \neq s \in \mathcal{W}_{d-1}$, suppose extra $d - \max \{ \rank(r), \rank(s) \}$ zeros are omitted from both, $r \succ_T s$ iff $F(\nu, r, s) \equiv f_\nu(s) - f_\nu(r) > 0$ for all $\nu \in (-\infty, \infty)$. This can also be expressed in terms of power majorization~\cite{kribs2013}.

Markedly, all distinct spectra on an isoentropy surface $f_\nu(r)$ are \emph{mutually} strongly incomparable. Since assisted-comparability is dictated by whether the sign of $F(\nu, r, s)$ is consistent with respect to $\nu$. For $r \neq s$, $F(\nu, r, s)$ cannot vanish for all $\nu$. Any discrepancy would include $F(\nu, r, s) = 0$, leading to strong incomparability.

Furthermore, any pair of strongly incomparable spectra is on some isoentropy surface $f_\nu(r)$. When $\rank(r) = \rank(s)$, $F(\nu, r, s)$ is $\nu$-continuous~\cite{klimesh2007}, any change in $\sign [F(\nu, r, s)]$ implies $F(\nu, r, s) = 0$. When $\rank(r) < \rank(s)$, $F(\nu \leq 0, r, s)$ are nonexistent, similarly, $r \nsim_T s$ are isoentropic by $f_{\nu>0}(r)$. 

Hence, the aim of catalysis is to catalyze incomparable spectra that are \emph{not} isoentropic by generalized entropies.

We find that a rigorous measure of a resource concerned with majorization should have it that equivalue states with distinct spectra are incomparable. If this is not fulfilled then such states are not truly equals as conversions between them are possible.

\begin{figure}[t]
	\centering
	\includegraphics[height=0.244\textwidth]{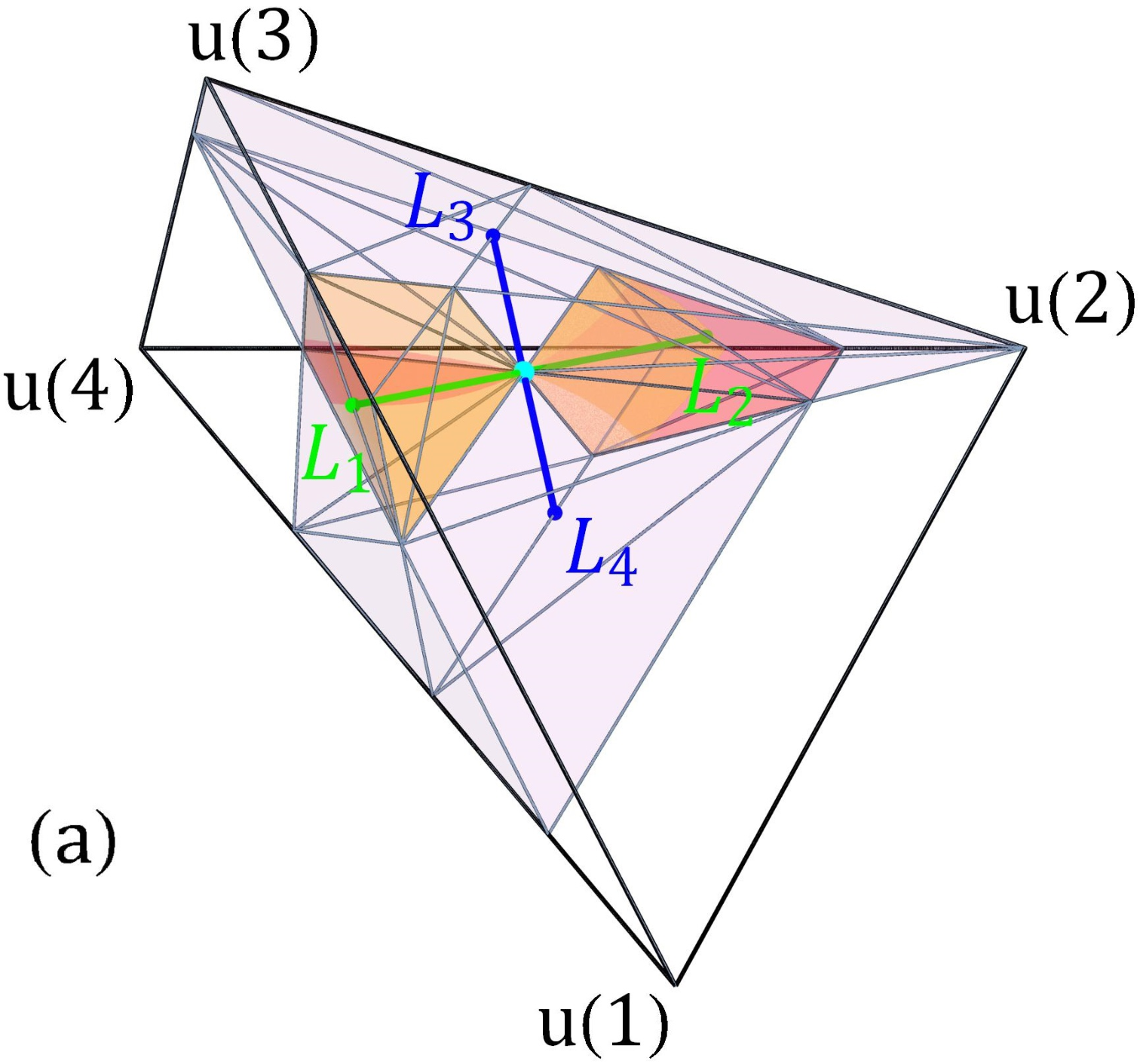}
	\includegraphics[width=0.202\textwidth]{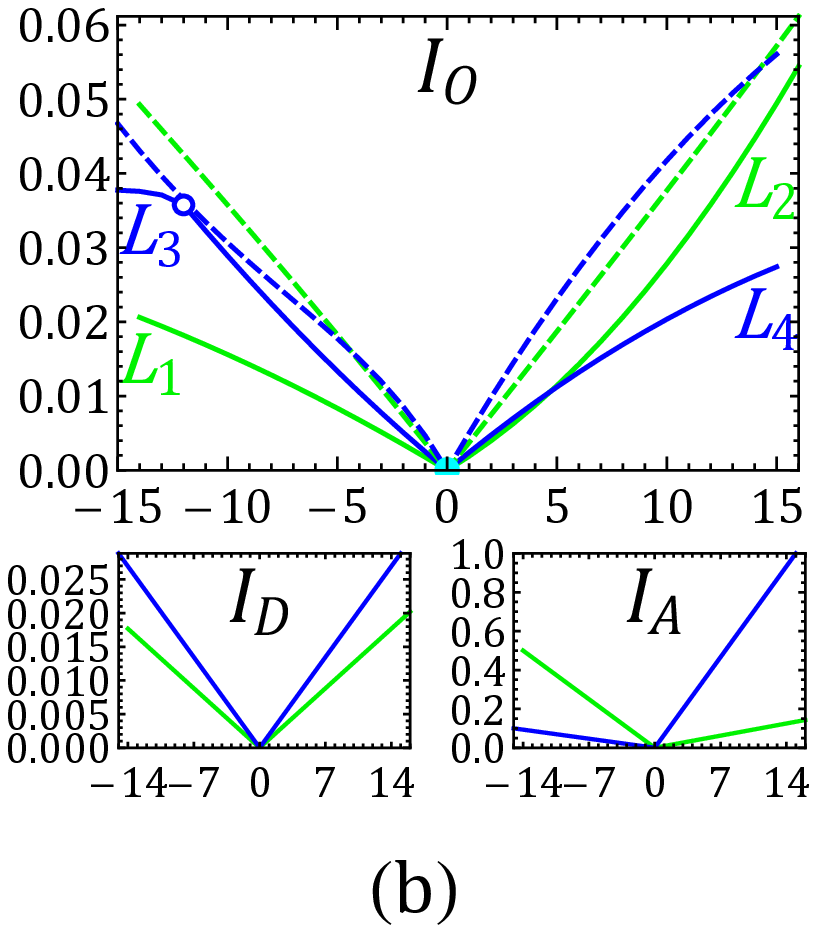}
	\caption{\label{fig:example} (a) Weyl chamber $\mathcal{W}_3$. $\eta = (0.5, 0.3, 0.15, 0.05)$ (cyan). Colored regions $\mathcal{I}_3(\eta)$: $n_\mathcal{I}(\eta) = 1$ (purple), $n_\mathcal{I}(\eta) = 2$ noncatalyzable (red), $n_\mathcal{I}(\eta) = 2$ catalyzable (orange). Red and orange regions are demarcated by $f_0(\eta)$ and $f_1(\eta)$ for either branch. (b) Comparison of $I_\mathrm{(l)O}(\xi, \eta)$, $I_\mathrm{D}(\xi, \eta)$, $I_\mathrm{A}(\xi, \eta)$. $\xi \in L_{1(3)} L_{2(4)}$ [green (blue) lines], horizontal axes reflect their natural length. $I_{\mathrm{lO}}$ (dashed lines). $\tilde{\xi}$ (blank point).}
\end{figure}

\textit{Inversion rank.}\textemdash We introduce the following notion. 
\begin{definition}
	For $r \nsim s$, examining their inequalities in sequence disregarding intermediate equalities, we call each index $l$ where the order reverses an inversion index. The totality of such is named the inversion rank $n_\mathcal{I}(r, s)$.
\end{definition}
We note $n_\mathcal{I}(r, s) \leq d-2$ (excluding $i = 1, d$); when comparable, $n_\mathcal{I}(r, s) = 0$.

Incomparability can be classified by inversion rank and further distinguished by the position and direction of the inversion. We denote for spectral sets incomparable to $r$ by the signs of majorization inequalities $\sign \{ A_j(r') - A_j(r) \}$ written as region vectors. For example in FIG.~\ref{fig:example}(a) we have: $n_\mathcal{I}(\eta) = 0$ comprises $\mathcal{M}_3^-(\eta) = (-0, -0, -0, 0)_\eta$, $\mathcal{M}_3^+(\eta) = (0+, 0+, 0+, 0)_\eta$; $n_\mathcal{I}(\eta) = 1$ comprises $(-, \dot{0}+, \dot{0}+, 0)_\eta$, $(-\dot{0}, -\dot{0}, +, 0)_\eta$, $(+, -\dot{0}, -\dot{0}, 0)_\eta$, $(\dot{0}+, \dot{0}+, -, 0)_\eta$; $n_\mathcal{I}(\eta) = 2$ comprises $(-, + , -, 0)_\eta$, $(+, - , +, 0)_\eta$. Note ``$-0$'' (``$0+$'') signifies inclusion of a boundary, ``$\dot{0}$'' indicates mutually exclusive boundaries, ending ``$0$'' upholds final equality.

We present a condition for strongly incomparability using inversion rank contrapositioning a result in Ref.~\cite{jonathan1999}.

\begin{remark}
	If $n_\mathcal{I}(r, s)$ is an odd number then $r \nsim_T s$.
\end{remark}
\begin{proof}
	Odd $n_\mathcal{I}(r, s)$ means an odd traversal of the Lorenz curves giving $r_1 > s_1$ and $r_d > s_d$. For any $d'$-catalyst $c$, also $r_1 c_1 > s_1 c_1$ and $r_d c_{d'} > s_d c_{d'}$. Thus $A_{11}(r \otimes c) > A_{11}(s \otimes c)$ and $A_{dd'-1}(r \otimes c) < A_{dd'-1}(s \otimes c)$, $r \nsim_T s$.
\end{proof}

The following knowledge is geometrically relevant.

\begin{remark}
	Any set indicated by a region vector is convex.
\end{remark}
\begin{proof}
	Let $t = (1 - p) r + p r'$, if $A_j(r) - A_j(s) > 0$ and $A_j(r') - A_j(s) > 0$ then $A_j(t) - A_j(s) = (1 - p) [A_j(r) - A_j(s)] + p [A_j(r') - A_j(s)] > 0$. Likewise for all cases. 
\end{proof}

\textit{Majorization relations.}\textemdash In this segment we develop some convenient majorization relations.

We raise a family of spectra with nonfixed dimensions, which has been used as a standard for nonuniformity~\cite{gour2015}.

\begin{proposition}
	Let $u(k)$ be a $k$-uniform spectrum,
	\begin{equation}
	\begin{array}{cc}
		u_i(k) = \frac{1}{k} \gamma_i^k, & k \in \mathcal{Z}^+,
	\end{array}
	\label{eq:descend}
	\end{equation}
	where $\gamma_i^k = 1$ ($i \leq k$), $0$ ($i > k$).
	Given $r \in \mathcal{W}_{d-1}$, for a specified $k$ the following hold:
	(1) $r \succ u$ iff $\rank(r) \leq \rank(u)$;
	(2) $u \succ r$ iff $r_1 \leq u_1$;
	(3) $r \nsim u$ iff $\rank(r) > \rank(u)$ and $r_1 > u_1$.
\end{proposition}
\begin{proof}
	(1) More plainly, $u(k) \in \mathcal{W}_{d-1}$ are of the form and order: $(1, 0, \ldots, 0) \succ (\frac{1}{2}, \frac{1}{2}, 0, \ldots, 0) \succ \ldots \succ (\frac{1}{d - 1}, \ldots, \frac{1}{d - 1}, 0) \succ (\frac{1}{d}, \ldots, \frac{1}{d})$. Clearly, they correspond to maximally mixed states for their respective embedded dimensions, the claim becomes apparent. (2) $u_{i \leq k}(k)$ are uniform while $r_i$ are nonincreasing. (3) When the prior conditions are unmet, $r \nsim u$.
\end{proof}

We see $\mathcal{W}_{d - 1}$ is the convex hull of the extreme points $\{u(k)\}_1^d$ [FIG.~\ref{fig:example}(a)]. As any convex sum of $\{u(k)\}_1^d$ is nonincreasing, and any $r \in \mathcal{W}_{d - 1}$ can be decomposed as: $r = \sum_{n=1}^{d} p_k u(k)$, where $\frac{1}{d} p_d = r_d$, $\frac{1}{k} p_k = r_k - r_{k + 1}, \, k \in \{ 1, \ldots, d - 1 \}$, $(p_k) \in \Delta_{d - 1}$. Additionally, not only is $u(k)$ useless as a catalyst~\cite{jonathan1999}, but also noncatalyzable in related conversions since for $r \nsim u$, provably $n_\mathcal{I}(r, u) = 1$.

We now detail a spectral family constructed from $u(k)$.

\begin{proposition}\label{prop:2}
	Let $\chi(q)$ be the depolarized pure spectrum of a $d$-system,
	\begin{equation}
		\chi(q) = \xoverline[0.9]{u(1)u(d)} \equiv (1 - q)u(1) + qu(d),
		\label{eq:depolar}
	\end{equation}
	where $q \in [0, 1]$. Given $r \in \mathcal{W}_{d-1}$, for a specified $q$ the following hold:
	(1) $r \succ \chi$ iff $r_1 \geq \chi_1$;
	(2) $\chi \succ r$ iff $r_d \geq \chi_d$;
	(3) $r \nsim \chi$ iff $r_1 < \chi_1$ and $r_d < \chi_d$.
\end{proposition}
\begin{proof}
	(1) Necessity is self-evident, we give proof of sufficiency. Let $r_1 \geq \chi_1$, assuming $r \nsim \chi$, then there exists an inversion index $l < d$ where $A_j(r) \geq A_j(\chi), j \in \{1,\ldots,l-1\}$ and $A_l(r) < A_l(\chi)$ or $A_{l + 1}'(r) > A_{l + 1}'(\chi)$. The inversion requires $r_l < \chi_l$, and since $r_i$ are nonincreasing while $\chi_{i > 1}$ are strictly uniform, $A_{l + 1}'(r) < A_{l + 1}'(\chi)$ contradicting the above. Thus $r$ is comparable to $\chi$. When $r_1 > \chi_1$ it must be $r \succ \chi$. When $r_1 = \chi_1$, if $ \chi \succ r$ then $\chi_2 \geq r_2$, due to uniformity of $\chi_{i > 1}$ the only possibility is $r = \chi$ where reflectively $r \succ \chi$. (2) For $\chi \succ r$, equivalently $A_j'(r) \geq A_j'(\chi), j \in \{ 1, \dots, d \}$ with equality for $j = 1$. Necessarily $r_d \geq \chi_d$ whence the other inequalities are implied by uniformity of $\chi_{i > 1}$. (3) When the prior conditions are unmet, $r \nsim \chi$.
\end{proof}

Also, $\chi(q)$ is noncatalyzable since for $r \nsim \chi$, again $n_\mathcal{I}(r, \chi) = 1$. We now derive a majorization criterion implying a general range for incomparability.

\begin{corollary}
	Let $r, s \in \mathcal{W}_{d-1}$, if $r_1 \geq \overline{\chi}_1[s] = 1 + (1-d) s_d$ then $r \succ s$, if $r_d \geq \underline{\chi}_d[s] = \frac{1 - s_1}{d-1}$ then $s \succ r$.
\end{corollary}
\begin{proof}
	Through Proposition~\ref{prop:2} we can find for $r$ its upper and lower pure depolarization bound $\overline{\chi}[s] \equiv \chi(\overline{q} = d s_d)$ and $\underline{\chi}[s] \equiv \chi(\underline{q} = d\frac{1 - s_1}{d - 1})$ respectively such that $\overline{\chi}[s] \succeq s \succeq \underline{\chi}[s]$. When $q \in (\underline{q}, \overline{q})$, $\chi(q) \nsim s$ [Fig.~\ref{fig:lorenz}(b)]. The corollary ensues after a repetition of Proposition~\ref{prop:2} by identifying spectrum $r$ satisfying $r \succ \overline{\chi}[s]$ or $\underline{\chi}[s] \succ r$.
\end{proof}

\textit{Inconvertibility.}\textemdash Nielsen conjectured that the probability of picking at random a pair of incomparable spectra in a $d$-system tends to $1$ as $d \rightarrow \infty$~\cite{nielsen1999}. It was established that densely many of these are in fact strongly incomparable~\cite{clifton2002}. Our objective is to find for a single spectrum its ratio of incomparable spectra in the Weyl chamber. We refer to this as inconvertibility. The higher its value the less likely it is to engage in majorization.

\begin{definition}
	Let $r \in \mathcal{W}_{d-1}$, its inconvertibility $C_d(r)$ is defined as its incomparable spectral volume $V[\mathcal{I}_{d - 1}(r)]$ divided by the volume of the entire chamber $V[\mathcal{W}_{d - 1}]$:
	\begin{equation}
		C_d(r) = \frac{V[\mathcal{I}_{d - 1}(r)]}{V[\mathcal{W}_{d - 1}]}.
	\end{equation}
\end{definition}

Technically, for $V[\mathcal{I}_{d - 1}(r)]$ we calculate the volumes of $\mathcal{M}_{d - 1}^\mp(r)$ by means of computational geometry and subtract them from the whole, i.e., $V[\mathcal{I}_{d - 1}(r)] = V[\mathcal{W}_{d - 1}] - V[\mathcal{M}_{d - 1}^-(r)] - V[\mathcal{M}_{d - 1}^+(r)]$. For $\rank(r) < d$, $\mathcal{M}_{d - 1}^+(r)$ lies on the boundaries of $\mathcal{W}_{d - 1}$ and is negligible. $u(k)$ and $\chi(q)$ can be directly integrated.

We find $C_d[u(d - 1)] = 1 - (d - 1)^{(1 - d)}$. We assess $u(d - 1)$ maximizes $C_d$ since its value tends to $1$ as $d \rightarrow \infty$. Hence, maximal inconvertibility may occur very close to complete disorder which is most convertible.

We choose $\chi(q)$ as a probe of inconvertibility since it spans the entire entropy range and has manifest majorization relations [FIG.~\ref{fig:complex}(a)]. The figure agrees with the intuition of statistical complexity~\cite{lopez-ruiz1995}, affirming the use of $C_d$ as an indicator of complexity~\cite{seitz2016}.

\begin{figure}[t]
	\centering
	\includegraphics[width=0.230\textwidth]{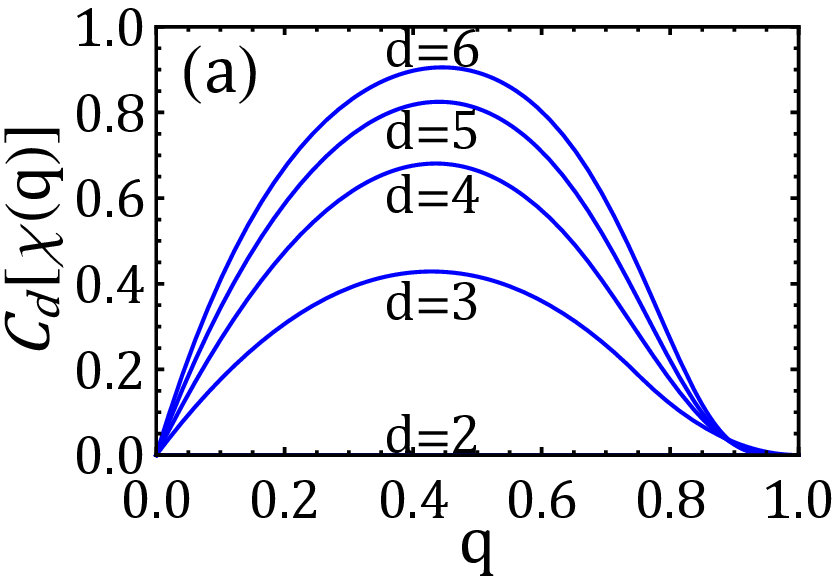}
	\includegraphics[width=0.247\textwidth]{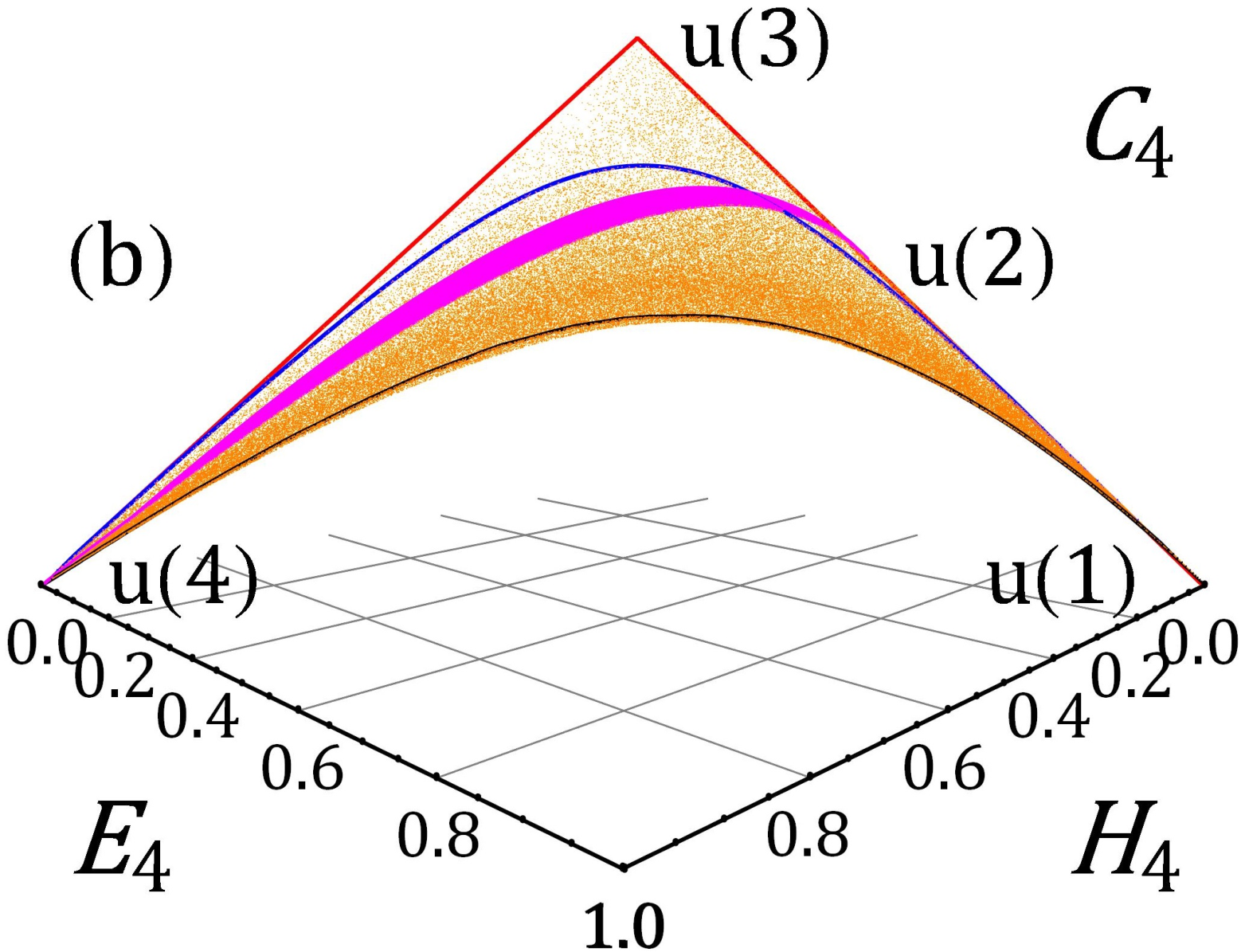}
	\caption{\label{fig:complex} (a) Inconvertibility of $\chi(q)$ for $d \in \{ 2, \ldots, 6 \}$. $C_d[\chi(q)]$ tapers to zero at the extremities and vanishes for $d = 2$. It increases with $d$, as $C_d^\mathrm{m}[\chi(q^\mathrm{m} \rightarrow 0.5)] \rightarrow 1$. Near $q = 1$, the curves cross at different points. (b) Complementary relation of $C_4(r) = 1 - E_4(r) - H_4(r)$ depicted by random points (orange). $u(3)$ sits expectedly at the apex. The edges of $\mathcal{W}_3$, excepting $\chi(q)$ (blue) and $\xoverline[0.9]{u(2)u(4)}$ (magenta) lie on the upperbound (red). In contrast, normalizing the Boltzmann $N$-partition maximally inconvertible by entropy~\cite{seitz2016} yields the spectral analog: $(1 - \frac{k}{N}) u(1) + \frac{k}{N} u(k), k \in \{1,\ldots,N\}$. Withal, $\xoverline[0.9]{u(2)u(4)}$ is an oscillatory band expanding at midrange. Again assuming a $2 \times 2$-system, $\xoverline[0.9]{u(1)u(2)}$ and $\xoverline[0.9]{u(2)u(4)}$ upperbound the product spectral surface $r_1 (1 - r_1 - r_2 - r_3) = r_2 r_3$, the lowerbound of which is close to that of the entire chamber.}
\end{figure}

We believe $V[\mathcal{M}_{d - 1}^\mp(r)]$ are natural reflections of disorder and disequilibrium~\cite{lopez-ruiz1995}. They can be supplemental to related measures by imposing a finer hierarchy, e.g., for isoentropic spectra, greater $V[\mathcal{M}_{d - 1}^+(r)]$ means more parties majorizing $r$, which signifies greater disorder and more utility for say its associated entangled (coherent) pure state. To inquire further, the complementary relation $C_4(r) = 1 - E_4(r) - H_4(r)$ where $E_4(r) = V[\mathcal{M}_3^-(r)] / V[\mathcal{W}_3]$ and $H_4(r) = V[\mathcal{M}_3^+(r)] / V[\mathcal{W}_3]$ is shown by random points in $\mathcal{W}_3$ [FIG.~\ref{fig:complex}(b)]. We surmise that the central region of $C_d$ would compress upwards for greater $d$ with $\chi(q)$ touching the upperbound.

\textit{Conclusion.}\textemdash Throughout this treatise we have studied the circumstance of incomparability in majorization comprehensively under a quantum setting. The rudiments of our theory can be accommodated to more specific contexts. For resource theories it can be adapted to examine the lack of resources inhibiting state conversion. Analysis can thence be conducted on the relations between resource deficiency, infidelity and failure rate. The mutual strong incomparability of generalized isoentropic spectra is an intrinsic trait of entropies underlying the difficulty of related processes. It can also be worthwhile to see what categorization by inversion rank may reveal. The spectral geometry inside a Weyl chamber entails an ulterior layer of attributes pertaining to a state, by which a bridge has been extended to complexity theory via inconvertibility. It may be of interest to know what are the most convertible states by entropy. All in all, we hope to have contributed meaningfully to the discussion.

\bibliography{ciqi}

\end{document}